\documentclass{article}

\usepackage{PRIMEarxiv}

\usepackage[utf8]{inputenc} 
\usepackage[T1]{fontenc}    
\usepackage{hyperref}       
\usepackage{amsmath}
\usepackage{url}            
\usepackage{booktabs}       
\usepackage{amsfonts}       
\usepackage{nicefrac}       
\usepackage{microtype}      
\usepackage{lipsum}
\usepackage{fancyhdr}       
\usepackage{graphicx}       
\graphicspath{{media/}}     
\usepackage{tikz}
\usetikzlibrary{calc,positioning,decorations.pathreplacing}
\usepackage{algorithm}
\usepackage{algpseudocode}
\usepackage{amssymb}

\newtheorem{Definition}{Definition}
\newtheorem{Theorem}{Theorem}
\newtheorem{Assumption}{Assumption}
\newtheorem{Lemma}{Lemma}
\newtheorem{Remark}{Remark}
\newtheorem{Corollary}{Corollary}
\newtheorem{proof}{Proof}
\title{Robustness of Online Identification-based Policy Iteration to Noisy Data
}

\author{
  Bowen Song, Andrea Iannelli \\
  Institute for Systems Theory and Automatic Control \\
  University of Stuttgart \\
  Stuttgart\\
  \texttt{\{bowen.song, andrea.iannelli\}@ist.uni-stuttgart.de} \\
}

\begin{document}
\maketitle

\begin{abstract}
This article investigates the core mechanisms of indirect data-driven control for unknown systems, focusing on the application of policy iteration (PI) within the context of the linear quadratic regulator (LQR) optimal control problem. Specifically, we consider a setting where data is collected sequentially from a linear system subject to exogenous process noise, and is then used to refine estimates of the optimal control policy. We integrate recursive least squares (RLS) for online model estimation within a certainty-equivalent framework, and employ PI to iteratively update the control policy. In this work, we investigate first the convergence behavior of RLS under two different models of adversarial noise, namely point-wise and energy bounded noise, and then we provide a closed-loop analysis of the combined model identification and control design process. This iterative scheme is formulated as an algorithmic dynamical system consisting of the feedback interconnection between two algorithms expressed as discrete-time systems. This system theoretic viewpoint on indirect data-driven control allows us to establish convergence guarantees to the optimal controller in the face of uncertainty caused by noisy data. Simulations illustrate the theoretical results.
\end{abstract}


\section{Introduction} 
Data-driven control is a very active area of research aimed at developing control strategies for systems where a precise mathematical model is unavailable, a scenario increasingly common in complex modern applications. This field encompasses a wide range of approaches with different problem settings, techniques, and objectives. It is beyond the scope of this Introduction to review them all, and we refer the reader to the following works and references therein \cite{HOU20133,khaki2023introduction,10156081,Doerfler_CSM_23_Pt2,FAULWASSER202392,annurev:/content/journals/10.1146/annurev-control-030323-024328}. One direction closely related to this work is indirect data-driven control, where data is collected first to estimate a system model, which is then used inside model-based control methods \cite{articlesimulation,8732482,chatzikiriakos2024endtoendguaranteesindirectdatadriven,pmlr-v178-tsiamis22a}. This approach makes use of system identification \cite{ljung1999system} and, in cases where the controller is updated during operation, is related to indirect adaptive control \cite{AC} or model-based reinforcement learning \cite{annurev:/content/journals/10.1146/annurev-control-062922-090153}. By blending data-driven insights with established model-based strategies, indirect data-driven control offers a flexible framework for tackling control problems in complex, dynamic environments.

In this article, we focus on a classic problem of increasing importance  in the optimal control and reinforcement learning communities: policy iteration (PI) for solving the linear quadratic regulator (LQR) problem. PI is a dynamic programming algorithm for optimal control \cite{lewis2012optimal,Bertsekas2022_Abstr_DP} that plays a foundational role in approximate dynamic programming and reinforcement learning algorithms \cite{park2020structured,9794431,9444823,BERTSEKAS_Newton_method_22,ECC}.
The PI algorithm consists of two main steps—policy evaluation and policy improvement—both of which traditionally rely on an accurate model of the plant. In the standard formulation, convergence to the optimal policy is guaranteed under certain assumptions about the cost and system's dynamics \cite{doi:10.1080/00207179.2015.1079737}.

The LQR problem is a foundational optimal control problem frequently used as a benchmark to compare data-driven control approaches and, owing to its tractability, analytically understand their fundamental properties \cite{articlesimulation,9444823,ECC,9691800,10383604,IJRNC}. For example, \cite{9444823} investigates the impact of additive uncertainties in model-based PI for continuous-time LQR, while \cite{ECC} examines the robustness of PI in the presence of parameter uncertainties. {The regret analysis of system identification and LQR algorithms has been investigated from a statistical learning perspective in \cite{pmlr-v119-cassel20a, pmlr-v19-abbasi-yadkori11a, pmlr-v119-simchowitz20a}.} In \cite{10383604}, a data-driven policy gradient method that integrates recursive least squares (RLS) with a model-based policy gradient approach is proposed, with convergence analyzed using averaging theory, {and in \cite{10383516}, an adaptive control framework is proposed for LQR problem}. Additionally, the authors' previous work \cite{IJRNC} compared indirect and direct data-driven PI for LQR and provided their advantages and disadvantages via theoretical analysis. {System-theoretic tools \cite{10540567} were employed for analysis in \cite{10383604, IJRNC,10383516}. The works discussed earlier \cite{10383604, IJRNC,pmlr-v119-cassel20a, pmlr-v19-abbasi-yadkori11a, pmlr-v119-simchowitz20a} adopt an indirect data-driven control framework, which involves system identification followed by controller design based on the estimated system dynamics. In contrast, the direct data-driven control framework bypasses the system identification step and directly optimizes the control policy. In \cite{pmlr-v80-fazel18a, 10091214}, direct data-driven policy gradient methods leveraging zeroth-order optimization were proposed for the noise-free discrete-time and continuous-time LQR problem. Similarly, \cite{carnevale2024datadrivenlqrfinitetimeexperiments} introduced a data-driven policy gradient method incorporating a novel zeroth-order gradient estimation technique for the noise-free LQR problem. More recently, \cite{zhao2024dataenabledpolicyoptimizationdirect} proposed a direct adaptive data-driven policy gradient method to handle LQR with noise.} 

In this study, we develop an indirect data-driven policy iteration approach to solve the LQR problem for an unknown system subject to additive adversarial process noise. Specifically, we consider the twofold scenario where the noise is point-wise bounded and energy bounded. We begin by examining the convergence properties of RLS identification, providing a finite-sample analysis that extends existing asymptotic convergence results for RLS \cite{bruce2020convergence}. Our analysis is meaningful for providing guarantees in indirect data-driven control that employ RLS for online system identification.
Then, we consider the feedback interconnection between the RLS algorithm and the PI scheme, where the gain matrix is refined iteratively through PI steps that use model estimates generated by RLS from online noisy data. By leveraging an algorithmic dynamical systems viewpoint on this interconnection, we frame this iterative process as a nonlinear feedback interconnection and carry out a system theoretic closed-loop analysis. With these results, we establish the conditions under which the algorithmic system converges to the desired values (i.e., the optimal controller and the true system model) and, if convergence is not achieved, we provide a guaranteed upper bound on the suboptimal solution. Our analysis captures the noise in the online collected data as a source of disturbance, enabling an input-to-state stability result with an intuitive, practical interpretation. In contrast to previous studies, such as {\cite{10383604,IJRNC,10383516,pmlr-v80-fazel18a,10091214,carnevale2024datadrivenlqrfinitetimeexperiments}}, which assume noise-free data, our approach accommodates adversarial noise and relaxes assumptions necessary for closed-loop analysis compared to our previous work \cite{IJRNC}. The analysis in our work provides insights into the impact of noise within the indirect data-driven policy iteration framework. This work serves as an example for analyzing online concurrent learning and controller design algorithms, highlighting how noise influences convergence and control performance.

The main contributions of this work are summarized as follows:
\begin{itemize}
\item Convergence analysis of RLS under pointwise bounded noise and energy-bounded noise.
\item A system-theoretic analysis of the concurrent learning and controller design algorithm using noisy data.
\end{itemize}
The paper is organized as follows: Section \ref{sec:PSP} introduces the problem setting and provides essential preliminaries. Section \ref{RLSsystem} investigates the convergence properties of recursive least squares with adversarial noisy data. Section \ref{sec:main} details the methodologies of the indirect data-driven policy iteration and analyzes the convergence properties of the coupled RLS and PI system. Section \ref{sec:simulation} illustrates the theoretical findings. Finally, Section \ref{sec:conclusion} provides a concluding summary of the work.

\subsubsection*{Notations:}
We denote by $A\succeq 0$ and $A\succ0$ a positive semidefinite and positive definite matrix $A$, respectively. The symbol $\mathbb{S}^n_+$ represents the set of real $n\times n$ symmetric and positive semidefinite matrices. The sets of non-negative and positive integers are denoted by $\mathbb{Z}_+$ and $\mathbb{Z}_{++}$, respectively. The Kronecker product is represented as $\otimes$, and $vec(A)=[a_1^\top,a_2^\top,...,a_n^\top]^\top$ stacks the columns $a_i$ of matrix $A$ into a vector. The symbols $\lfloor x \rfloor$ and $\lceil x \rceil$ denote the floor function, which returns the greatest integer smaller or equal than $x\in\mathbb{R}$, and ceil function, which returns the smallest integer greater or equal than $x\in\mathbb{R}$, respectively. 
For matrices and vectors, $\lvert \cdot\rvert$ denotes their Frobenius and Euclidean norm, respectively. A function belongs to class $\mathcal{K}$ if it is continuous, strictly increasing, and vanishing at the origin. A function $\beta(x , t)$ is called $\mathcal{KL}$ function if $\beta(x , t)$ decreases to 0 as $t\rightarrow0$ for every $x\geq0$ and $\beta(\cdot , t)\in \mathcal{K}$ for all $t\geq 0$. For $Y \in \mathbb{R}^{m \times n}$ and $r>0$, we define $\mathcal{B}_{r}(Y) := \{X \in \mathbb{R}^{m \times n} {\big|} ~|X-Y| < r\}$. We consider generic sequences $\{Y_t\}$ as maps $\mathbb{Z}_+\rightarrow \mathbb{R}^{m \times n}$, and we denote by $\lVert Y \rVert_\infty:=\sup\limits_{t\in{\mathbb{Z}_+}}|Y_t|$ and $\lVert Y \rVert_2:=\sum\limits_{t=0}^{\infty}|Y_t|$.

\section{Problem Setting and Preliminaries}\label{sec:PSP}
We consider discrete-time linear time-invariant (LTI) systems of the form
\begin{equation}\label{LTI}
  x_{t+1}=Ax_t+Bu_t+w_t,
\end{equation}
where $x_t\in \mathbb{R}^{n_x}$ is the system state, $u_t\in \mathbb{R}^{n_u}$ is the control input and $t$ denotes the timestep. The system matrices $(A,B)$ are unknown but assumed to be stabilizable, as is standard in data-driven control approaches \cite{9691800, pmlr-v99-tu19a}; $w_t\in \mathbb{R}^{n_x}$ represents the adversarial process noise acting on the system. 

In this work, we consider two models for the noise: 
\begin{itemize}
    \item point-wise bounded noise \cite{10644295}, where the noise sequence $\{w_t\}$ satisfies:
\begin{equation}\label{Instant}
        \lVert w \rVert_\infty \leq L_\infty, \quad L_\infty \in (0,\infty),
    \end{equation}
    where $L_\infty$ is an upper bound on the noise magnitude;
    \item energy bounded noise \cite{VENKATASUBRAMANIAN2024556}, where the noise sequence $\{w_t\}$ satisfies: 
     \begin{equation}\label{Energy}
        \lVert w \rVert_2\leq {L}_2, \quad L_2 \in (0,\infty),
    \end{equation}
    where $L_2$ represents the noise energy. This type of noise is a specific form of case of point-wise bounded noise with the additional property, the magnitude of the noise converges to zero quickly enough to be summable, implying $\lim\limits_{t\rightarrow \infty}\lvert w_t\rvert =0$, which offers advantages in certain control applications. 
\end{itemize}
    
The objective is to design a state-feedback controller $u_t=Kx_t$ that minimizes the following infinite horizon cost {for the noise-free plant}:
\begin{equation}\label{Cost}
  J(x_t,K)=\sum\limits_{k=t}^{+\infty} r(x_k,u_k)=\sum\limits_{k=t}^{+\infty} x_k^\top Qx_k+u_k^\top Ru_k,
\end{equation}
where $R\succ 0$ and $Q\succeq 0$. When a stabilizing gain $K$ is applied, ensuring that $A+BK$ is Schur stable, the corresponding cost $J(x_t, K)$ can be expressed in terms of the quadratic form $x_t^\top Px_t$. Here, $P\succ 0$ represents the quadratic kernel of the cost function associated with $K$ \cite{9691800}, which is determined by the model-based Bellman equation:
\begin{equation}\label{MBBE}
  P=Q+K^\top RK+(A+BK)^\top P(A+BK).
\end{equation} 
In optimal control theory \cite{lewis2012optimal}, it is well established that the solution to the linear quadratic regulator (LQR) problem is a linear state-feedback control. The optimal feedback gain $K^*$ is determined by:
\begin{subequations}\label{Kpolicyimprovement_EQ}
\begin{align}
K^*&=-(R+B^\top P^*B)^{-1}B^\top P^*A, \label{Kpolicyimprovement} \\
P^*&=Q+A^\top P^*A-A^\top P^*B(R+B^\top P^*B)^{-1}B^\top P^*A, \label{DARE} 
\end{align}
\end{subequations}
where $P^*$ is the quadratic kernel of the optimal cost (value function) and is the unique solution of the discrete algebraic Riccati equation (DARE) in \eqref{DARE}. The system of equations in \eqref{Kpolicyimprovement_EQ} provides a way to compute the optimal feedback gain $K^*$ that minimizes cost \eqref{Cost}.

\subsection{Policy Iteration} \label{2024secPI}
Even when the system model is known, directly solving DARE \eqref{DARE} can become computationally challenging for high-dimensional systems. Policy iteration (PI) offers an efficient, iterative method to compute the optimal gain $K^*$ by-passing this calculation. The fundamental model-based version of the PI algorithm \cite{1099755}, which requires knowledge of the system matrices $A$ and $B$, is summarized in Algorithm \ref{Algo1}.

\begin{algorithm}
  \caption{Model-based policy iteration.}\label{Algo1}
  \begin{algorithmic}
      \Require $A,B$, a stabilizing policy gain $K_0$ 
      \For{$i=0,1,...,+\infty$}
        \State \textbf{Policy Evaluation: find $P_{i}$} 
        \State \begin{equation}\label{2024PE}
            P_{i}=Q+K_i^\top RK_{i}+(A+BK_{i})^\top P_{i}(A+BK_{i})
        \end{equation}
        \State \textbf{Policy Improvement: update gain $K_{i+1}$}
        \State \begin{equation}\label{2024PI}
            K_{i+1}=-(R+B^\top P_iB)^{-1}B^\top P_iA
        \end{equation} 
      \EndFor
  \end{algorithmic}
\end{algorithm}

The key properties of Algorithm \ref{Algo1} are presented in the following theorem.
\begin{Theorem}{Properties of model-based PI {\cite{1099755}\cite[Theorem 4]{IJRNC}}} \label{theorem1}
 \\
If the system dynamics $(A,B)$ are stabilizable, and $K_0$ is stabilizing, then 
  \begin{enumerate}
    \item $P_0\succeq P_1 \succeq ... \succeq P^*$;
    \item $K_i$ stabilizes the system $(A,B)$, $\forall i\in \mathbb{Z}_+$;
    \item $\lim\limits_{i\rightarrow\infty}P_i=P^*$, $\lim\limits_{i\rightarrow\infty}K_i=K^*$;
    \item $\lvert P_{i+1}-P^* \rvert \leq c \lvert P_{i}-P^* \rvert$ with $c<1$, $\forall i \in \mathbb{Z}_+$.
  \end{enumerate}
\end{Theorem}
This theorem establishes that, under stabilizability of \eqref{LTI} and appropriate initialization of $K_0$, the sequence $\{P_i\}$ generated by policy iteration converges exponentially to the optimal solution $P^*$, with $K_i$ stabilizing the system at each iteration. Theorem \ref{theorem1} is a standard result on PI. However, leveraging a dynamical system viewpoint, we can obtain an additional results.
\subsubsection{PI System Analysis}
In \cite{ECC}, we investigated the convergence of PI algorithm with nominal system $(A,B)$ by equivalently reformulating it as a dynamical system. The main steps are as follows. Define the functions $\alpha(P_i):=B^\top P_iA$ and $\beta(P_i):=R+B^\top P_iB$, where $\beta(P_i)$ is a positive definite matrix and thus always invertible. By substituting the policy improvement step \eqref{2024PI} into the policy evaluation step \eqref{2024PE}, the relationship between $P_i$ and $P_{i+1}$ is given by:
\begin{equation}\label{relationPseq}
\begin{split}
   P_{i+1}&=Q+A^\top P_{i+1}A\\
   &~+\alpha (P_i)^\top\beta(P_i)^{-1}\beta(P_{i+1})\beta(P_i)^{-1}\alpha(P_i) \\
      &~-\alpha (P_{i+1})^\top\beta(P_i)^{-1}\alpha(P_i)\\
      &~-\alpha (P_{i})^\top\beta(P_i)^{-1}\alpha(P_{i+1}).
\end{split}
\end{equation}
Using the identity $vec(EFG)=(F^\top  \otimes E) vec(G)$ from \cite{Petersen2008} and defining 
\begin{equation}\label{Gmma}
    \Gamma(P_i):=Q+\alpha (P_i)^\top\beta(P_i)^{-1}R\beta(P_i)^{-1}\alpha(P_i),
\end{equation}
we can rewrite \eqref{relationPseq} as:
\begin{equation}\label{rerelationPseq2}
\begin{split}
      \mathcal{A}(P_i)vec(P_{i+1})=vec\left(\Gamma(P_i)\right),
  \end{split}
\end{equation}
where $\mathcal{A}(P_i):=I_{n_x}\otimes I_{n_x}-\Omega(P_i) \otimes\Omega(P_i)$ and $\Omega(P_i):=A^\top -\alpha (P_i)^\top\beta(P_i)^{-1}B^\top$. \\
If $\mathcal{A}(P_i)$ is invertible, we have:
\begin{equation}\label{ISS31}
\begin{split}
    &vec(P_{i+1})=\mathcal{A}(P_i)^{-1} vec\left(\Gamma(P_i)\right).
\end{split}
\end{equation}
The transformation from \eqref{relationPseq} to \eqref{ISS31} involves reshaping the vectorized terms back into a square matrix, thereby establishing the iterative relationship between $P_{i+1}$ and $P_i$. This process can be formalized as:
\begin{equation}\label{ISS3}  
\begin{split}
P_{i+1}=\mathcal{L}^{-1}_{(A,B,P_i)}\left(\Gamma(P_i)\right).
\end{split}
\end{equation}
where $\mathcal{L}^{-1}_{(\cdot)}(\cdot)$ is an operator that reconstructs the matrix $P_{i+1}$ using $(A,B)$ and $P_i$. 

{This formulation allows the sequence $\{P_i\}$ obtained from Algorithm \ref{Algo1} to be interpreted as a discrete-time dynamical system, abstracting the PI algorithm into an algorithmic dynamic and enabling the analysis of its convergence properties, which serves as the foundation for the subsequent analysis.}
To this aim, the invertibility of $\mathcal{A}(P_i)$ must be ensured. According to Theorem \ref{theorem1}, when $P_i \succeq P^*$, the invertibility of $\mathcal{A}(P_i)$ is guaranteed. This condition yields convergence of \eqref{ISS3} to $P^*$, as established in Theorem \ref{theorem1}.


Additionally, in \cite[Theorem 4]{ECC}, we explored an alternative condition that guarantees the invertibility of $\mathcal{A}(P_i)$ and ensures exponential convergence, without relying on the well-known condition $P_i \succeq P^*$, as discussed in Theorem \ref{theorem1}.
\begin{Theorem}[Exponential convergence of PI \cite{ECC}]\label{Theorem2}
There exists a constant $\delta_1 > 0$, such that for any $P_i\in \mathcal{B}_{\delta_1}(P^*)$, $\mathcal{A}(P_i)$ is invertible and the following inequality holds:
\begin{equation}
    \lvert P_{i+1}-P^* \rvert \leq \sigma \lvert P_{i}-P^* \rvert, \quad \quad\forall i \in \mathbb{Z}_+,
\end{equation}
where $\sigma \in (0,1)$. 
\end{Theorem}
The advantage of Theorem \ref{Theorem2} is to guarantee the existence of a region around the optimal $P^*$ such that if $P_0$ is initialized there, the sequence $\{P_i\}$ generated by PI guarantees the invertibility of $\mathcal{A}(P_i)$. Figure \ref{fig:enter-label} illustrates the region where $P \succeq P^*$  as the shaded area, indicating where convergence is guaranteed by Theorem \ref{theorem1}. The remaining region, depicted within the circle, represents the area where convergence is ensured by Theorem \ref{Theorem2}.
\begin{figure}[H]
    \centering
    \begin{tikzpicture}[scale=0.85, transform shape]
    \draw[->] (0,0) -- (4,0) node[below] {$\lambda_1(P)$}; 
    \draw[->] (0,0) -- (0,4) node[left] {$\lambda_2(P)$};  
    \draw[thick] (2,2) -- (2,4) ;
    \draw[thick] (2,2) -- (4,2) ;
    \filldraw[gray, opacity=0.2] (2,2) rectangle (3.9,3.9);
       \draw[opacity=0.8, draw=black, thick] (2,2) circle [radius=1.5];
    \filldraw [black] (2,2) circle (2pt) node[above right] {$P^*$};
    \node at (1.9,1.1) {\textcolor{black}{\fontsize{8}{14}\selectfont $\mathcal{B}_{\delta_1}(P^*)$}};
    \node at (3,3.7) {\textcolor{black}{\fontsize{8}{14}\selectfont $P\succeq P^*$}};
\end{tikzpicture} 
    \caption{2-dimensional Graphic Representation}
    \label{fig:enter-label}
\end{figure}
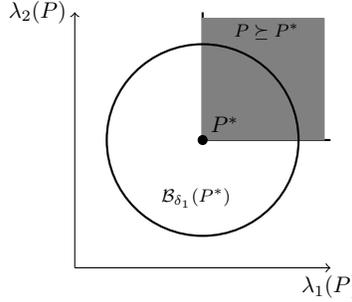
\subsection{Recursive Least Squares} 
When the system dynamics are unknown, least squares identification is a possible strategy to identify the model parameters. We can rewrite system \eqref{LTI} as: 
\begin{equation}\label{ConvergenceSystemID1}
x_{t+1}=Ax_t+Bu_t+w_t=\underbrace{\left[A~B\right]}_{=:\theta}\underbrace{\left[\begin{array}{cc}
         x_t  \\
         u_t 
    \end{array}\right]}_{=:d_t}+w_t.
\end{equation}
Given a dataset $\{d_k,x_{k+1}\}^T_{k=1}$ collected over a trajectory of length $T$, an estimate $\hat{\theta}$ of system matrix $\theta$ can be obtained by minimizing the least-squares loss function \cite{ljung1999system}:
\begin{equation}\label{Solution}
  \begin{split}
     \theta \in\mathop{\mathrm{arg~min}}\limits_{\hat{\theta}} \sum\limits_{k=1}^{T}(x_{k+1}-\hat{\theta} d_k)^\top(x_{k+1}-\hat{\theta} d_k).
  \end{split}
\end{equation}
When the matrix $H_T:=\left(\sum\limits_{k=1}^{T}d_kd_k^\top\right)$ is invertible, $\hat{\theta}$ has a closed-form solution:
\begin{equation}\label{Solution2}
  \begin{split}
     \hat{\theta}=\left(\sum\limits_{k=1}^{T}x_{k+1}d_k^\top\right)H_T^{-1}.
  \end{split}
\end{equation}

This (batch) least squares approach estimates the parameters in a single step, utilizing all data points at once. In contrast, the recursive least squares (RLS) algorithm is particularly valuable for online estimation scenarios \cite{AC}, whereby estimates are incrementally updated as new data becomes available. Defining the estimated system matrix at time $t$ as $\hat{\theta}_t:=[\hat{A}_t,\hat{B}_t]$, the RLS algorithm update equations, are given as follows and summarized in Algorithm \ref{Algo4}.

\begin{subequations}\label{rls111}
    \begin{align}
    H_{t}&=H_{t-1}+d_td_t^\top,\\
        \hat{\theta}_t&=\hat{\theta}_{t-1}+(x_{t+1}-\hat{\theta}_{t-1}d_{t})d^\top_{t}H_t^{-1}.        
    \end{align}
\end{subequations}

\begin{algorithm}[H]
  \caption{Recursive least squares.}\label{Algo4}
  \begin{algorithmic}
      \Require An initial estimate of the system dynamic $\hat{\theta}_0$ and $H_0 \succ 0$
      \For{$t=1,...,\infty$}
          \State \textbf{Given $\{x_{t+1},d_t\}$}
          \State $H_{t}=H_{t-1}+d_td_t^\top$
          \State $\hat{\theta}_t=\hat{\theta}_{t-1}+(x_{t+1}-\hat{\theta}_{t-1}d_{t})d^\top_{t}H_t^{-1}$
      \EndFor
  \end{algorithmic}
\end{algorithm}

In the context of RLS, it is essential to quantify the time-varying estimation error, denoted as $\Delta \theta_t:=\hat{\theta}_t-\theta$, which evolves as data are collected over time. This term arises due to the initial estimation error and the effect of the noise models \eqref{Instant} and \eqref{Energy}. Using Algorithm \ref{Algo4}, we can derive the recursive expression for the estimation error $\Delta \theta_t$ as follows:
\begin{equation}\label{error}
\begin{split}
    \Delta \theta_t&=\hat{\theta}_{t-1} H_{t-1}H_t^{-1}+(\theta d_t+w_t)d_t^\top H_t^{-1}-\theta H_tH_t^{-1}\\
    &=(\hat{\theta}_{t-1}-\theta) H_{t-1}H_t^{-1}+w_td_t^\top H_t^{-1}\\
    &=(\hat{\theta}_{t-2}-\theta) H_{t-2}H_t^{-1}+(w_td_t^\top+w_{t-1}d_{t-1}^\top )H_t^{-1}\\
    &=(\hat{\theta}_{0}-\theta) H_{0}H_t^{-1}+\left(\sum_{k=1}^{t}w_kd_k^\top\right) H_t^{-1}.\\
\end{split}
\end{equation}

In the derivations above, the first equality uses the RLS update equation in \eqref{rls111} and the last equality is obtained by the recursively applying \eqref{error}. In the next section, we will analyze how the estimation error behaves under the presence of adversarial noisy data. 

\section{Recursive Least Squares with Adversarial Noise Data}\label{RLSsystem}
Before analyzing the property of RLS with noisy data, we first recall a property of data sequence $\{d_t\}$, where $d_t$ is defined in \eqref{ConvergenceSystemID1}, which plays a crucial role in ensuring the convergence of the RLS estimator. This property, named local persistency, captures the excitation level of the data sequence.
\begin{Definition}{Local persistency \cite[Definition 2]{IJRNC}}\label{def2}\\
    A sequence $\{Y_t\} \in \mathbb{S}^{n}_+$ is locally persistent if there exist $N\geq 1, M\geq 1$ and $\alpha > 0$ such that , for all $j=Mk+1$ with $k\in\mathbb{Z}_+$,
    \begin{equation}\label{ConvergenceSystemid8}
        \sum\limits_{t=0}^{N-1}Y_{t+j} \succeq \alpha I_n.
    \end{equation}
    The numbers $\alpha$ and $N$ are respectively, the lower bound and persistency window of $\{Y_t\}$. $M$ is the persistency interval. A sequence $\{Y_t\} \in \mathbb{R}^{n\times m}$ is locally persistent if $\{Y_t Y_t^\top\}$ is locally persistent.
\end{Definition}
The concept of local persistency was first introduced in our previous work \cite{IJRNC} as a relaxed condition of the persistency condition from in \cite{bruce2020convergence}. Local persistency provides a sufficient condition for the convergence of the RLS algorithm with noise-free data, as demonstrated in \cite[Theorem 2]{IJRNC}. We introduce the following assumption which holds throughout the work.
\begin{Assumption}\label{asspersistentence}
    The data sequence $\{d_t\}$ is locally persistent with parameters $N=N_d,~M=M_d~\mathrm{and}~\alpha=\alpha_d$.
\end{Assumption}
Assumption \ref{asspersistentence} can be met by appropriately selecting the excitation signal $u_t$. In a later section, we will address how to design $u_t$ to satisfy this assumption. Building on this assumption, we now extend the analysis to include the convergence of RLS under the influence of adversarial noise. To facilitate this, we introduce an additional assumption regarding the data sequence:
\begin{Assumption}[Boundedness of data sequence $\{d_t\}$] \label{ass1}
    The data sequence $\{d_t\}$ satisfies: 
    \begin{equation}\label{boundD}
        \lVert d \rVert_\infty \leq \Bar{d},
    \end{equation}
    where $\Bar{d} \in (0, \infty)$ is a constant.
\end{Assumption}
Because of the boundedness of the noise sequence $\{w_t\}$, Assumption \ref{ass1} is guaranteed if we apply a stabilizing gain $K$. Having established these preliminaries, we now present the following theorem that analyzes the convergence properties of the RLS estimation error in the presence of bounded noisy data.
\begin{Theorem}\label{theoremRLS}
   If Assumptions \ref{asspersistentence} and \ref{ass1} are satisfied and the noise satisfies \eqref{Instant}, then the estimation error of RLS initialized with $\hat{\theta}_0$ and $H_0=aI(a>0)$ is bounded by:
        \begin{equation}\label{eqtheorem2}
    \lvert \hat{\theta}_t -\theta \rvert \leq \beta_\theta(\lvert \hat{\theta}_0 - \theta \rvert,t) +\gamma_\theta(\lVert w \rVert_\infty), \quad \forall t\in \mathbb{Z}_{++}
    \end{equation}
    where $\beta_\theta(\lvert \hat{\theta}_0 - \theta \rvert,t):=\frac{a(M_d+N_d)\lvert \hat{\theta}_0 - \theta \rvert}{\min(a,\alpha_d)t}$; $\gamma_\theta(x):=\bar{d}\eta x$; $\bar{d}$ is defined in \eqref{boundD}; $\eta:=\frac{(n_x+n_u)(M_d+N_d)}{\min(a,\alpha_d)}$.
\end{Theorem}
The proof of Theorem \ref{theoremRLS} is provided in Appendix \ref{ProofofTheorem4}. The result of Theorem \ref{theoremRLS} can be interpreted as an input-to-state stability (ISS) result \cite{Sontag2008, JIANG2001857}. The function $\beta_\theta(\cdot,\cdot)$ is a $\mathcal{KL}$ function, representing the error due to initialization $\hat{\theta}_0$, which decreases to zero as $t$ approaches infinity. The function $\gamma_\theta(\cdot)$ is a $\mathcal{K}$ function, capturing the error introduced by the noise term $w_t$. This function is non-zero unless $\lVert w \rVert_\infty=0$. Based on Theorem \ref{theoremRLS}, we can derive the following corollary, which is a standard corollary of ISS results.     
\begin{Corollary}\label{Coro2}
  Using the assumptions and notations of Theorem \ref{theoremRLS}, if $\lim\limits_{t\rightarrow \infty} \lvert w_t\rvert=0$, then we have $\lim\limits_{t\rightarrow \infty} \lvert \hat{\theta}_t -\theta \rvert=0$.
\end{Corollary}
The proof of Corollary \ref{Coro2} closely follows the steps outlined in \cite[Appendix D3]{IJRNC} and is omitted here. As discussed in Section \ref{sec:PSP}, the energy-bounded noise condition in \eqref{Energy} represents a particular case of \eqref{Instant}, where $\lim\limits_{t\rightarrow \infty} \lvert w_t\rvert=0$. Thus, Theorem \ref{theoremRLS} and Corollary \ref{Coro2} are applicable. However, by directly using \eqref{Energy}, a stronger result than those provided in \eqref{eqtheorem2} and Corollary \ref{Coro2} can be achieved.
\begin{Corollary}[RLS with energy bounded noisy data]\label{Coro1}
Using the assumptions and notations of Theorem \ref{theoremRLS}, if the noise is energy bounded, i.e. satisfying \eqref{Energy}, the estimation error of RLS is bounded by:
\begin{equation}
            \lvert \hat{\theta}_t -\theta \rvert \leq \beta_\theta(\lvert \hat{\theta}_0 - \theta \rvert,t) +\beta_{e}(\lVert w \rVert_2,\sqrt{t}),~\forall t\in \mathbb{Z}_{++},
\end{equation}
with $\beta_e(\lVert w \rVert_2,\sqrt{t}):=\bar{d}\eta\frac{\lVert w \rVert_2}{\sqrt{t}}$.
\end{Corollary}
The proof of Corollary \ref{Coro1} is provided in Appendix \ref{ProofofCorollary2}. According to the corollary, the estimation error is bounded by two $\mathcal{KL}$-function. As $t$ approaches infinity, the estimation error converges to zero, which recovers with the result in Corollary \ref{Coro2}.

The analysis in this section provides analytical insight into the role of noise in RLS, illustrating how noise affects estimation accuracy and convergence. These results can be integrated with robust control techniques to guarantee the performance of indirect data-driven control employing online RLS algorithms.

Before concluding our discussion on RLS, we quantify the maximum estimation error of RLS for point-wise bounded noise, which can be derived from Theorem \ref{theoremRLS} as:
\begin{equation}\label{max1}
    \overline{\Delta\theta}(\hat{\theta}_0,\bar{d}):=\max\{\lvert \hat{\theta}_0-\theta\rvert, \beta_\theta(\lvert \hat{\theta}_0 - \theta \rvert,1) +\bar{d}\eta\lVert w \rVert_\infty\}. 
\end{equation}
The first term in \eqref{max1} represents the estimation error determined by the initialization at $t=0$, and the second term is the upper bound provided by Theorem \ref{theoremRLS} for $t\geq 1$.
This quantity can be interpreted as the largest estimation error for all $t\in \mathbb{Z}_+$, i.e. $|\Delta \theta_t|\leq \overline{\Delta\theta}(\hat{\theta}_0,\bar{d})$, and it is determined by the initialization $\hat{\theta}_0$ and the upper bound on the data sequence $\bar{d}$ defined in \eqref{boundD}. Similarly, for the energy bounded noise satisfying \eqref{Energy}, we can derive the maximum estimation error from Corollary \ref{Coro2} as:
\begin{equation}\label{max2}
    \overline{\Delta\theta}_e(\hat{\theta}_0,\bar{d}):=\max\{\lvert \hat{\theta}_0-\theta\rvert, \beta_\theta(\lvert \hat{\theta}_0 - \theta \rvert,1) +\bar{d}\eta\lVert w \rVert_2\}. 
\end{equation}
\section{Online Identification-based Policy Iteration}\label{sec:main}
In this section, we analyze the online identification-based policy iteration (ORLS+PI), which integrates the model-based PI from Algorithm \ref{Algo1} with the RLS algorithm presented in Algorithm \ref{Algo4}. This approach offers a practical solution for performing policy iteration in scenarios where the system dynamics are unknown. By concurrently optimizing the policy and conducting online system identification, the algorithm aims to improve the control performance iteratively. Our primary focus is to investigate the convergence properties and limitations of this combined approach from a system-theoretic perspective and its robustness to noise. 
\subsection{Algorithm Definition}
For the ORLS+PI algorithm, we collect the data sequence $\{d_t\}$ online with the control input $u_t$ given as:
\begin{equation}\label{2024ut}
    u_t=\hat{K}_tx_t+e_t,
\end{equation}
where $\hat{K}_t$ is the feedback gain and $e_t$ is a potentially non-zero feedforward term.
\begin{Remark}[Remark on $\hat{K}_t$]\label{RemarkK}
    The gain $\hat{K}_t$ in \eqref{2024ut} is referred to the on-policy gain \cite{10383604}, meaning that data are generated using the policy currently being updated. In this case, the $\hat{K}_t$ is generated by ORLS+PI algorithm. However, as discussed in \cite[Section 5.4]{IJRNC}, one advantage of indirect data-driven policy iteration is that the excitation can be also performed off-policy. i.e. the data can be generated using a different stabilizing policy $K$, that is not updated by the algorithmic dynamics.  
\end{Remark}
\begin{Remark}[Remark on $e_t$]
    The term $e_t$ represents an additional degree of freedom of the online policy, which can be used, for example, as an exploratory signal that explores the system in a random or targeted way \cite{8732482,9099297}. The purpose of including $e_t$ is to ensure the local persistency of the data sequence $\{d_t\}$, i.e. Assumption \ref{asspersistentence}. However, it is important to note that the subsequent analysis is agnostic to the specific choice of $e_t$.
    In this work, we assume that the sequence of the signal $\{e_t\}$ is bounded, i.e.
          \begin{equation}\label{boundede}
       \lVert e \rVert_\infty \leq \bar{e}.
    \end{equation}
    where $\bar{e}\in (0,\infty)$ is a constant that represents the upper bound of the signal magnitude at each timestep. 
\end{Remark}
The ORLS+PI algorithm involves at each iteration $t$ the following steps: 
\begin{itemize}
  \item Given a policy gain $\hat{K}_{t}$, which either originates from the initialization $(t=1)$ or the previous timestep, the cost function kernel estimate $\hat{P}_t$ is computed by solving the model-based Bellman equation \eqref{MBBE} using the current system estimates $\left(\hat{A}_{t-1},\hat{B}_{t-1}\right)$:
      \begin{equation}\label{IPI2a}
      \begin{split}
            \hat{P}_{t}=Q+\hat{K}_t^\top R\hat{K}_{t}&+\\
                  \left(\hat{A}_{t-1}+\hat{B}_{t-1}\hat{K}_{t}\right)^\top&\hat{P}_{t}\left(\hat{A}_{t-1}+\hat{B}_{t-1}\hat{K}_{t}\right).
      \end{split}
      \end{equation}
      
  \item  The physical system is excited with the control input $u_t$ introduced in \eqref{2024ut}. The state-input data $\{x_t,u_t,x_{t+1}\}$ is then used to recursively update the system dynamics estimates $\left(\hat{A}_t,\hat{B}_t\right)$ using the RLS Algorithm:
      \begin{subequations}\label{procedure}
        \begin{align}
  H_{t} &=H_{t-1}+d_td_t^\top,\label{IPI2b} \\
  \hat{\theta}_t &= \left(\hat{\theta}_{t-1}H_{t-1}+ x_{t+1}d_t^\top \right)H_{t} ^{-1}.\label{IPI2c}
         \end{align}
      \end{subequations}
  \item Using the updated estimates $\left(\hat{A}_t,\hat{B}_t\right)$, the policy is improved by solving for the new feedback gain $\hat{K}_{t+1}$:
  \begin{equation}\label{IPI2d}
    \hat{K}_{t+1}=-\left(R+\hat{B}_{t}^\top\hat{P}_{t}\hat{B}_{t}\right)^{-1}\hat{B}_{t}^\top\hat{P}_{t}\hat{A}_{t}.
  \end{equation}
\end{itemize}
To ensure the feasibility of the ORLS+PI algorithm, particularly regarding equations \eqref{IPI2a} and \eqref{IPI2d}, we will provide a detailed discussion on this topic in a later section. 

\begin{Remark}[Timestep $t$]
In this work, we use a single index $t$ for both the RLS estimate update and the PI policy update. While, in principle, each update could be tracked by its own independent index. The analysis in this section can be extended to handle different timescales for each update, following the approach outlined in \cite{IJRNC}.
\end{Remark}
The ORLS+PI algorithm is summarized in Algorithm \ref{Algo2} and is depicted in Figure \ref{fig:stucture} through a block diagram that emphasizes the dynamic viewpoint leveraged in this work. The closed-loop system, consisting of the physical system and the controller, is connected by the solid black lines in the figure and is subject to the exogenous noise term $w_t$. The algorithmic dynamics, formed by the PI and RLS algorithms, is placed inside the bottom shaded area and its interconnections are depicted by the dashed black lines.
\begin{algorithm}[H]
  \caption{ORLS+PI Algorithm}\label{Algo2}
  \begin{algorithmic}
      \Require $\hat{A}_0,\hat{B}_0,H_0$, the initial optimal policy gain $\hat{K}_1$ for system $(\hat{A}_0,\hat{B}_0)$
      \For{$t=1,...,\infty$}
        \State \textbf{Policy Evaluation: find $\hat{P}_{t}$ by \eqref{IPI2a}} 
        \State \State \textbf{Excite the system with $u_t=\hat{K}_tx_{t}+e_{t}$}
        \State \State \textbf{Collect the data $\gets (x_{t},u_{t},x_{t+1})$}
        \State \State \textbf{Use RLS in Algorithm \ref{Algo4} to update $\hat{A}_{t},\hat{B}_{t}$}  
        \State \textbf{Policy Improvement: update gain $\hat{K}_{t}$ by \eqref{IPI2d}}
        \State \EndFor
  \end{algorithmic}
\end{algorithm}
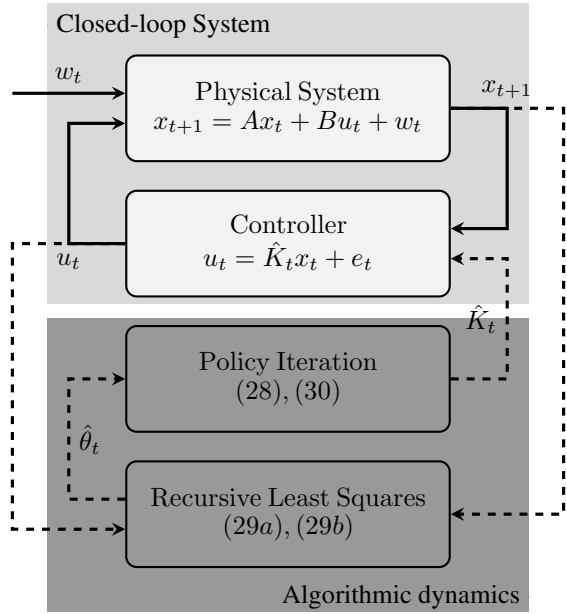
\begin{figure}[H]
    \centering
    \begin{tikzpicture}[auto, node distance=1.8cm]
        \tikzstyle{block} = [rectangle, thick, draw=black!100, minimum width=4.3cm, minimum height=1.5cm,fill=gray!10, text centered, rounded corners, minimum height=4em]
        \tikzstyle{blockblue} = [rectangle, thick, draw=black!100, minimum width=4.3cm, minimum height=1.5cm, text centered, rounded corners, minimum height=4em]
        \tikzstyle{line} = [draw]
        \tikzstyle{arrowblack}=[->, >=stealth, very thick, black]
        \tikzstyle{arrowblue}=[->, dashed,>=stealth, very thick, black]

        \fill[gray!30, opacity=0.4] (-3.2,-2.6) rectangle (3.2,1.4);

        \fill[gray!80, opacity=0.5] (-3.2,-6.7) rectangle (3.2,-2.8);
        \node [block] (block1) {$\begin{array}{cc}\mathrm{Physical~System}\\ x_{t+1}=Ax_t+Bu_t+w_t \end{array}$};
        \node [block, below of=block1] (block2) {$\begin{array}{cc}\mathrm{Controller}\\ u_{t}=\hat{K}_tx_t+e_t \end{array}$};
        \node [blockblue, below of=block2] (block3) {$\begin{array}{cc}\mathrm{Policy~Iteration}\\ (28),(30) \end{array}$};
        \node [blockblue, below of=block3] (block4) {$\begin{array}{cc}\mathrm{Recursive~Least~Squares}\\ (29a),(29b) \end{array}$};
        \filldraw[black] (-3.2, 1.1) circle (0pt) node[below, right]{Closed-loop System};
        \filldraw[black] (3.2, -6.5) circle (0pt) node[above,left]{Algorithmic dynamics};
        \draw[arrowblack] ($(block1.west)+(-1.5,+0.2)$)--node[above,black]{$w_t$}  ($(block1.west)+(0,+0.2)$);
        \draw[arrowblack] (block1.east) -- ($(block1.east)+(0.75,0)$) -- ($(block2.east)+(0.75,0.2)$) -- ($(block2.east)+(0,0.2)$);
        \draw[arrowblack] (block2.west) -- ($(block2.west)+(-0.75,0)$) -- ($(block1.west)+(-0.75,-0.2)$) -- ($(block1.west)+(0,-0.2)$);
        \draw[arrowblue] (block1.east) --node[above,black]{$x_{t+1}$} ($(block1.east)+(1.5,0)$) -- ($(block4.east)+(1.5,0)$) -- (block4.east);
        \draw[arrowblue] ($(block4.west)+(0,0.2)$) -- ($(block4.west)+(-0.75,0.2)$) -- node[right,black]{$\hat{\theta}_t$}($(block3.west)+(-0.75,0)$) -- (block3.west);
        \draw[arrowblue] (block3.east) -- ($(block3.east)+(0.75,0)$) --node[left,black]{$\hat{K}_t$} ($(block2.east)+(0.75,-0.2)$) -- ($(block2.east)+(0,-0.2)$);
        \draw[arrowblue] (block2.west) --node[below,black]{$u_{t}$} ($(block2.west)+(-1.5,0)$) -- ($(block4.west)+(-1.5,-0.2)$) -- ($(block4.west)+(0,-0.2)$);
    \end{tikzpicture}
    \caption{Concurrent identification and policy iteration scheme}
    \label{fig:stucture}
\end{figure}

\subsection{Convergence Analysis of ORLS+PI Algorithm}
As illustrated in Figure \ref{fig:stucture}, the dynamics of the policy iteration (PI) and recursive least-squares (RLS) can be analyzed as a feedback interconnection of two coupled dynamical systems. In the "system PI", the inputs are the estimates $\left(\hat{A}_t,\hat{B}_t\right)$ obtained from the RLS, and the dynamics are described by \eqref{IPI2d} and \eqref{IPI2a}. In the "system RLS", the inputs are the data $\{d_t\}$ and $\{x_{t+1}\}$ collected online from the physical system and perturbed by the noise, with the dynamics described by \eqref{IPI2b} and \eqref{IPI2c}.

The properties of "system PI" were recalled in Section \ref{2024secPI} and the properties of "system RLS" were investigated in Section \ref{RLSsystem}, which provides insight into the behavior of the RLS algorithm under adversarial noise conditions. To facilitate our analysis, we introduce the following notations: 
\begin{subequations}\label{29102024}
\begin{align}
\hat{\alpha}_t&:=\hat{B}_t^\top \hat{P}_t\hat{A}_t  \\
\hat{\beta}_t&=\hat{\beta}_t^\top:=R+\hat{B}_t^\top \hat{P}_t\hat{B}_t.
\end{align}
\end{subequations}
Before stating the main result, we introduce the following assumption.

\begin{Assumption}\label{assumption 2}
   The estimates $\left(\hat{A}_t,\hat{B}_t\right)$ obtained from RLS are stabilizable $\forall t\in \mathbb{Z}_+$. Given a stabilizable estimate $\left(\hat{A}_t,\hat{B}_t\right)$, we assume that $\hat{P}_t \succeq P^*_{\left(\hat{A}_t,\hat{B}_t\right)}$ $\forall t\in \mathbb{Z}_+$, where $\hat{P}_t$ is obtained via \eqref{IPI2a} and $P^*_{\left(\hat{A}_t,\hat{B}_t\right)}$ is the quadratic kernel of the value function associated with $\left(\hat{A}_t,\hat{B}_t\right)$ and is calculated by solving \eqref{DARE}.
\end{Assumption}
Assumption \ref{assumption 2} is a direct translation in the online identification-based setting of the standard requirement for the formulation of policy iteration, (cf. Theorem \ref{theorem1}). For further discussions and details on how to realize this assumption, we refer to \cite[Assumption 1, Assumption 2]{IJRNC}. We are finally ready to state the main convergence and robustness result of Algorithm \ref{Algo2}.
\begin{Theorem}[ORLS+PI Analysis $1$]\label{Theorem5}
    If Assumption \ref{assumption 2} is satisfied, then the ORLS+PI system formulated by \eqref{IPI2a}-\eqref{IPI2d} admits the following equivalent dynamical system representation: 
  \begin{subequations}\label{IPI23}
  \begin{align}
  \begin{split}\label{IPI213c}
      \hat{\theta}_{t+1} &= \left(\hat{\theta}_{t}\left(H_0+\sum\limits_{k=1}^{t}d_kd_k^\top \right)\right.\\
      &\left.\quad \quad+ \sum\limits_{k=1}^{t}x_{t+1}d_t^\top \right)\left(H_0+\sum\limits_{k=1}^{t}d_kd_k^\top \right) ^{-1},
  \end{split}  \\ 
  \begin{split}\label{IPI213d} \hat{P}_{t+1}&=\mathcal{L}^{-1}_{\left(\hat{A}_t,\hat{B}_t,\hat{P}_t\right)}\left(Q+\hat{\alpha}_t^\top\hat{\beta}_t^{-1}R\hat{\beta}_t^{-1}\hat{\alpha}_t\right).
  \end{split}
  \end{align}
\end{subequations}
Additionally, if Assumptions \ref{asspersistentence} and \ref{ass1} are satisfied and the noise satisfies \eqref{Instant}, then with the initialization $H_0=aI(a>0)$ and arbitrary $\hat{\theta}_0$, the estimates $\hat{P}_t$ and $\hat{\theta}_t$ satisfy the following relationships for all $t\in \mathbb{Z}_{++}$: 
  \begin{subequations}\label{ISSCA_TOT1}
  \begin{align}
   \left \lvert \hat{P}_{t}- P^*\right \rvert &\leq \beta_c\left(\left \lvert \hat{P}_{0}- P^*\right \rvert,t\right)+\gamma_c\left(\left \lVert \Delta\theta \right \rVert_{\infty}\right),\label{ISSCA12}\\
    \lvert \hat{\theta}_t -\theta \rvert &\leq \beta_\theta(\lvert \hat{\theta}_0 - \theta \rvert,t) +\gamma_\theta(\lVert w \rVert_\infty),\label{ISSCA212}
  \end{align}
\end{subequations}
where
\begin{itemize}
    \item $\beta_c\left(\cdot, \cdot \right):=c^t\left \lvert \hat{P}_{0}- P^*\right \rvert$ is a $\mathcal{KL}$-function with $c \in (0,1)$ defined in Theorem \ref{theorem1};
    \item $\gamma_c\left(\left \lVert \cdot\right \rVert_{\infty}\right):=\frac{\bar{C}}{1-c}\left \lVert \cdot\right \rVert_{\infty}$ is a $\mathcal{K}$-function with constant $\bar{C}$ given in the proof \eqref{14112024};
    \item $\beta_\theta(\cdot,\cdot)$ and $\gamma_\theta(\cdot)$ are defined in Theorem \ref{theoremRLS}.
\end{itemize}
\end{Theorem}
The proof of Theorem \ref{Theorem5} is provided in Appendix \ref{ProofofTheorem5} and is the result of combining Theorem \ref{theoremRLS} with \cite[Theorem 6]{IJRNC}.

We observe here that, regarding Assumption \ref{ass1}, there is no guarantee that the stabilizing property of $\hat{K}_t$ will hold for the true system $(A,B)$. In the on-policy setting (see Remark \ref{RemarkK}), we cannot ensure the boundedness of the data sequence. However, as discussed in Remark \ref{RemarkK}, the excitation can be performed off-policy. 
Specifically, all the analyses still hold if the system is excited using a fixed pre-stabilizing gain $K$. In this off-policy case, Assumption \ref{ass1} can be guaranteed.

Theorem \ref{Theorem5} describes the convergence properties of the ORLS+PI algorithm for arbitrary initial $\hat{\theta}_0$. If an assumption on the maximum estimation error \eqref{max1}, which also depends on $\hat{\theta}_0$ is made, then Assumptions \ref{ass1} and \ref{assumption 2} are not anymore required.
\begin{Assumption}\label{assumption3}
    The maximum estimation error of RLS satisfies the following condition:
    \begin{equation}\label{condition}
        \overline{\Delta\theta} (\hat{\theta}_0,\bar{D})\leq \min \{\bar{a}_p,\bar{b}_p\},
    \end{equation}
    where $\bar{a}_p$ and $\bar{b}_p$ are constants defined in \eqref{251020242} (see Theorem \ref{TheoremRPI} in Appendix \ref{ProofofTheorem6}) and $\bar{D}$ is defined in \eqref{definitionmathcalD} (see Lemma \ref{boundD} in Appendix \ref{ProofofTheorem6}).
\end{Assumption}
The value of $\bar{D}$ is quantitatively determined by both the upper bound of the noise and the sequence $\{\hat{K}_t\}$ applied to the system. Assumption \ref{assumption3} requires that the maximum estimation error from RLS remains within acceptable limits. This can be used in conjunction with recent findings on the inherent robustness of PI with inexact models \cite{ECC} to show that Algorithm \ref{Algo4} converges under different assumptions than Theorem 4. Under Assumption \ref{assumption3}, we can derive the following theorem.
\begin{Theorem}[ORLS+PI Analysis $2$]\label{theorem6}
If Assumption \ref{assumption3} is satisfied {and the initial $\hat{K}_0$ is selected as the optimal gain calculated by solving \eqref{Kpolicyimprovement_EQ} using $(\hat{A}_0,\hat{B}_0,Q,R)$}, then the ORLS+PI algorithm formulated by \eqref{IPI2a}-\eqref{IPI2d} admits the equivalent dynamical system representation in \eqref{IPI23}. Additionally, if Assumption \ref{asspersistentence} is satisfied and the noise satisfies \eqref{Instant}, then with the initialization $H_0=aI(a>0)$ and an initial $\hat{\theta}_0$ satisfying Assumption \ref{assumption3}, the estimates $\hat{P}_t$ and $\hat{\theta}_t$ satisfy the following relationships for all $t\in \mathbb{Z}_{++}$: 
  \begin{subequations}\label{ISSCA_TOT}
  \begin{align}
   \left \lvert \hat{P}_{t}- P^*\right \rvert &\leq \beta_{\sigma}\left(\left \lvert \hat{P}_{0}- P^*\right \rvert,t\right)+\gamma_{\sigma}\left(\left \lVert \Delta\theta \right \rVert_{\infty}\right),\label{ISSCA}\\
    \lvert \hat{\theta}_t -\theta \rvert &\leq \beta_\theta(\lvert \hat{\theta}_0 - \theta \rvert,t) +\gamma_D(\lVert w \rVert_\infty),\label{ISSCA2}
  \end{align}
\end{subequations}
where: 
\begin{itemize}
    \item $\beta_{\sigma}\left(\cdot, \cdot \right):=\sigma^t\left \lvert \hat{P}_{0}- P^*\right \rvert$ is a $\mathcal{KL}$-function with $\sigma \in (0,1)$ defined in Theorem \ref{Theorem2};
    \item $\gamma_{\sigma}\left(\left \lVert \cdot\right \rVert_{\infty}\right):=\frac{\bar{p}_a+\bar{p}_b}{1-\sigma}\left \lVert \cdot\right \rVert_{\infty}$ is a $\mathcal{K}$-function with $\bar{p}_a$ and $\bar{p}_b$ given in the proof \eqref{251020242};
    \item $\beta_\theta(\cdot,\cdot)$ is defined in Theorem \ref{theoremRLS};
    \item $\gamma_D(\lVert\cdot\rVert_{\infty}):=c_D\lVert\cdot\rVert_{\infty}$ with $c_D:=\bar{D}\eta$; $\eta$ is defined in Theorem \ref{theoremRLS} and $\bar{D}$ is defined in \eqref{definitionmathcalD}.
\end{itemize}
\end{Theorem}
The proof of Theorem \ref{theorem6} is provided in Appendix \ref{ProofofTheorem6}. Here, we outline the main steps involved in the proof. The proof relies primarily on Theorem \ref{theoremRLS}, which establishes the convergence of the RLS under a bounded data sequence and point-wise bounded noise, and on \cite[Theorem 6]{ECC}, which describes the inherent robustness of PI. The proof proceeds as follows:
\begin{enumerate}
    \item \textbf{Condition on Initialization $\hat{\theta}_0$}: The inherent robustness of PI guarantees that $\hat{K}_t$ stabilizes the system for all $t\in \mathbb{Z}_+$. Because in addition we have point-wise bounded noise and control inputs, we determine the upper bounded of the sequence $\{d_t\}$ denoted by $\bar{D}$. Then we determine the necessary condition (Assumption \ref{assumption3}) to sure that Theorem \ref{TheoremRPI} holds;
    \item \textbf{PI inherent robustness}: Leveraging the robustness properties of PI from \cite[Theorem 6]{ECC}, we can directly establish inequality \eqref{ISSCA};
    \item \textbf{System stabilization and bounded data sequence}: We have shown that the data sequence $\{d_t\}$ is upper bounded by $\bar{D}$. This allows us to prove inequality \eqref{ISSCA2};
\end{enumerate}

\begin{Remark}[Comparison between Theorems \ref{Theorem5} and \ref{theorem6}]\label{R2}
Theorem \ref{Theorem5} extends the results of \cite[Theorem 6]{IJRNC} to case studies involving bounded noisy data. Theorem \ref{Theorem5} relies on Assumptions \ref{ass1} and \ref{assumption 2} to derive ISS results \eqref{IPI23}. These assumptions provide a result by imposing no restrictions on the initial condition $\hat{\theta}_0$ of RLS.

In contrast, Theorem \ref{theorem6} removes the Assumptions \ref{ass1} and \ref{assumption 2} by introducing a specific condition on initialization and the upper bound of the data sequence, which is partially influenced by the noise level, as defined in \eqref{condition}. This requirement ensures that the maximum estimation error stays within the level of inherent robustness of PI. Therefore, the results under Theorem \ref{theorem6} only hold when the estimation error is sufficiently small. 

As discussed earlier, for Theorem \ref{Theorem5}, we can only perform off-policy excitation during the online data collection to ensure the boundedness of the data sequence. However, in the case of Theorem \ref{theorem6}, the closed-loop stability of the physical system is guaranteed. Therefore, we can directly employ excitation with the on-policy gain $\hat{K}_t$.
\end{Remark}

\begin{Remark}[Remark to Assumption \ref{assumption3}]
    Assumption \ref{assumption3} cannot be directly verified, as we only know the existence of $\bar{a}_p$ and $\bar{b}_p$. However, from a system-theoretical perspective provided by Theorem \ref{theorem6}, we know that if the initial condition is close to the true system and the upper bound of the noise is small, the coupled system is input-to-state stable with respect to the upper bound of the noise and the estimation error of system matrices. Moreover, the on-policy gain ensures stability as stated in Remark \ref{R2}. In other words, compared to Theorem \ref{theorem6}, with better prior knowledge of the system matrices, fewer assumptions are required to guarantee the performance of concurrent learning and controller design procedure.
\end{Remark}

Based on Theorem \ref{Theorem5} and Theorem \ref{theorem6}, we can now derive the following corollaries that help interpret the two theorems in a more intuitive and practical manner.
\begin{Corollary}[Finite sample analysis of $\lvert \hat{P}_t-P^*\rvert$]\label{Coro4}
        Using the notations and assumptions of Theorem \ref{Theorem5} and given an iteration $t_\mathrm{re} > 1$, the distance between $\lvert \hat{P}_t-P^*\rvert$ can be quantified as:
    \begin{equation}
    \begin{split}
        \left \lvert \hat{P}_{t}- P^*\right \rvert &\leq \beta_c\left(\left \lvert \hat{P}_{t_\mathrm{re}}- P^*\right \rvert,t-t_\mathrm{re}\right)\\
        &~~+\gamma_c\left( \sup\limits_{k\geq t_\mathrm{re}}{\lvert \Delta \theta_k\rvert}\right), \quad\forall {t\geq t_\mathrm{re} };
    \end{split}
    \end{equation}
    Similarly, using similar notations and assumptions of Theorem \ref{theorem6}, we have:
    \begin{equation}
    \begin{split}
                 \left \lvert \hat{P}_{t}- P^*\right \rvert &\leq \beta_{\sigma}\left(\left \lvert \hat{P}_{t_\mathrm{re}}- P^*\right \rvert,t-t_\mathrm{re}\right)\\
                 &\quad+\gamma_{\sigma}\left(\sup\limits_{k\geq t_\mathrm{re}}{\lvert \Delta \theta_k\rvert}\right),\quad\forall {t\geq t_\mathrm{re} }.
    \end{split}
    \end{equation}
\end{Corollary}
The proof of Corollary \ref{Coro4} follows directly by reformulating the equations \eqref{ISSCA12} and \eqref{ISSCA}. 
\begin{Corollary}\label{Coro3}
    Under the conditions of Theorem \ref{Theorem5} and Theorem \ref{theorem6}, if  $\lim\limits_{t\rightarrow \infty} \lvert w_t\rvert=0$, then  $\lim\limits_{t\rightarrow \infty} \lvert \hat{\theta}_t -\theta \rvert=0$, $\lim\limits_{t\rightarrow \infty} \lvert \hat{P}_t -P^* \rvert=0$ and $\lim\limits_{t\rightarrow \infty} \lvert \hat{K}_t -K^* \rvert=0$.
\end{Corollary}
Corollary \ref{Coro3} is a standard corollary of ISS results and it can be proved for example by following the steps outlined in \cite[Appendix D3]{IJRNC}. From this corollary, if the data sequence is locally persistent and noise term $w_t$ vanished at infinity, $\{\hat{P}_t\}$ {obtained from the concurrent learning and controller design algorithm} converges asymptotically to the optimal $P^*$.

\begin{Corollary}[Energy bounded noise]\label{Coro5}
 Using the notations of Theorem \ref{theorem6}, for the energy bounded noise satisfying \eqref{Energy}, if 
 \begin{equation}
     \overline{\Delta\theta}_e(\hat{\theta}_0,\bar{D})\leq \min\{\bar{a}_p,\bar{b}_p\},
 \end{equation}
 where $\overline{\Delta\theta}_e(\cdot,\cdot)$ is defined in \eqref{max2}, then the ORLS+PI algorithm formulated by \eqref{IPI2a}-\eqref{IPI2d} admits the equivalent dynamical system representation in \eqref{IPI23}. If Assumption \ref{asspersistentence} is satisfied, 
the estimates $\hat{P}_t$ and $\hat{\theta}_t$ satisfy the following relationships for all $t\in \mathbb{Z}_{++}$:
  \begin{subequations}\label{ISSCA_TOT121212}
  \begin{align}
   \left \lvert \hat{P}_{t}- P^*\right \rvert &\leq \beta_{\sigma}\left(\left \lvert \hat{P}_{0}- P^*\right \rvert,t\right)+\gamma_{\sigma}\left(\left \lVert \Delta\theta \right \rVert_{\infty}\right),\label{ISSCA121212}\\
    \lvert \hat{\theta}_t -\theta \rvert &\leq \beta_\theta(\lvert \hat{\theta}_0 - \theta \rvert,t) +\beta_{D}(\lVert w \rVert_2,\sqrt{t}),\label{ISSCA2eg1}
  \end{align}
\end{subequations}
where $\beta_D(\lVert w \rVert_2,\sqrt{t}):=\bar{D}\eta\frac{\lVert w \rVert_2}{\sqrt{t}}$ and $\eta$ is defined in Theorem \ref{theoremRLS}. 
\end{Corollary}
Corollary \ref{Coro5} can be proved by integrating the results of Corollary \ref{Coro1} with Theorem \ref{theorem6}. Based on \eqref{ISSCA_TOT121212}, we can recover the asymptotic results stated in Corollary \ref{Coro4}. A similar corollary for Theorem \ref{Theorem5} can also be derived by combining Corollary \ref{Coro1} is omitted here. 

In this work and previous \cite{IJRNC}, we analyze the ORLS+PI algorithm as a dynamical system and provide input-to-state stability (ISS) results to characterize the closed-loop behavior. In \cite{IJRNC}, we focused on noise-free data, considering the persistency level of the data sequence. In contrast, this work accounts for bounded noisy data and assumes that the sequence is locally persistent. The analysis in this section provides a mathematical description of how adversarial noise impacts the performance of the ORLS+PI algorithm. This insight enables us to characterize the conditions under which noise affects estimation accuracy and convergence, informing guidelines for robust algorithm initialization and parameter tuning in noisy environments.
\section{Simulations}\label{sec:simulation}
In this section, we present simulation results\footnote{The Matlab codes used to generate these results are accessible from the repository: \href{https://github.com/col-tasas/2024-SysIDbasedPIwithNoisyData}{https://github.com/col-tasas/2024-SysIDbasedPIwithNoisyData}} to illustrate some of the properties of online identification-based policy iteration discussed in the previous sections.
\subsection{Comparison between different types of noise}
We consider the following system which was already used in prior studies \cite{articlesimulation,9691800,IJRNC}: 
\begin{equation}\label{LTIsimulation}
  x_{t+1}=\underbrace{\left[\begin{array}{ccc}
            1.01 & 0.01 & 0 \\
            0.01 & 1.01 & 0.01 \\
            0 & 0.01 & 1.01 
          \end{array}\right]}_A x_t+\underbrace{\left[\begin{array}{ccc}
            1 & 0 & 0 \\
            0 & 1 & 0 \\
            0 & 0 & 1 
          \end{array}\right]}_B u_t+w_t.
\end{equation}

The weight matrices $Q$ and $R$ are set to $0.001I_3$ and $I_3$, respectively. The initial estimates for system matrices $A$ and $B$ are set as:
\begin{equation}
\begin{split}
    \hat{A}_0=A+0.5I_3,\\
    \hat{B}_0=B+0.5I_3.
\end{split}
\end{equation}
The matrix $H_0$ for RLS is initialized as $0.1 I_6$. The initial stabilizing policy gain $\hat{K}_0$ is set to the optimal LQR gain associated with $(\hat{A}_0, \hat{B}_0, Q, R)$. The dithering signal $e_t$ of the policy \eqref{2024ut} is distributed uniformly with each entry sampled independently from the interval $[-10,10]$. Figure \ref{fig:simulation} illustrates the convergence of the quadratic kernel of the value function $\hat{P}_t$, representing the closed-loop evaluation of the cost function with the feedback gain $\hat{K}_t$ under different noise conditions, which are set as:
\begin{subequations}\label{noiseset}
    \begin{align}
        \mathrm{PB1}:\lvert w_t \rvert&=\frac{0.5}{t}+0.5;\\
         \mathrm{PB2}:\lvert w_t \rvert&=\frac{0.5}{t};\\
        \mathrm{EB}:\lvert w_t \rvert&=\frac{0.5}{t^2}.
    \end{align}
\end{subequations}
\begin{figure}[H]
    \centering
    \includegraphics[width=0.5\linewidth]{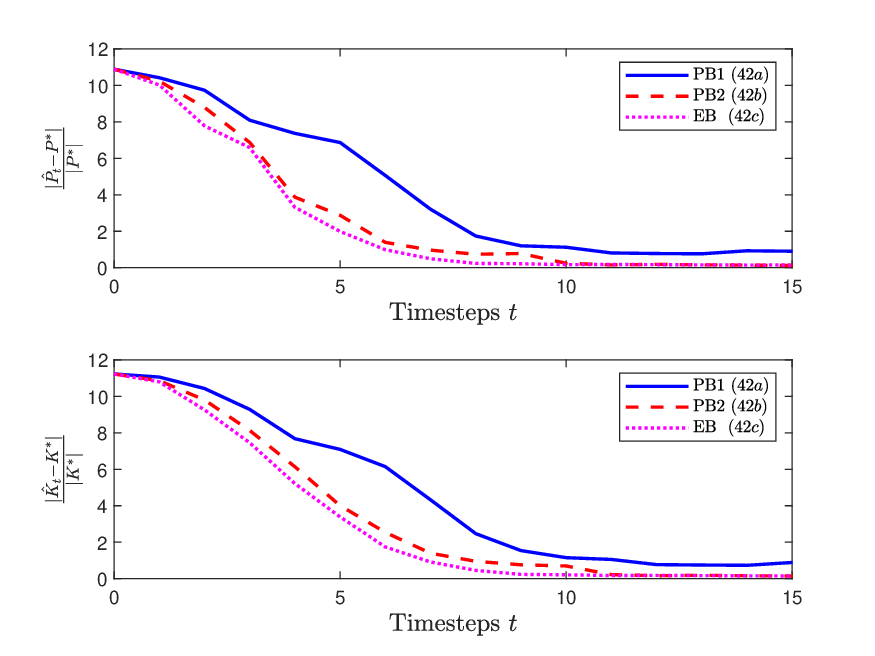}
    \caption{Comparisons of convergence behaviors of ORLS+PI with different types of noise}
    \label{fig:simulation}
\end{figure}

The blue solid line shows the convergence under point-wise bounded noise. This setup results in a non-vanishing error between $\hat{P}_t(\hat{K}_t)$ and ${P}^*({K}^*)$ due to the persistent noise component, as discussed in Corollary \ref{Coro4}. The red dashed line uses noise vanishing as $t\rightarrow\infty$ but is not energy bounded. This condition yields convergence to the optimal values, as detailed in Corollary \ref{Coro3}. The magenta dotted line shows energy bounded noise. This configuration, in line with Corollary \ref{Coro5}, achieves convergence to the optimal values. 
\subsection{Comparison between Policy Iteration and Policy Gradient}
We compare our OLRS+PI algorithm with a recently proposed method that combines online RLS with a model-based policy gradient approach \cite{10383604}, referred to here as ORLS+PG. The system dynamics $(A,B)$ and the weight matrices $Q$ and $R$ are set according to the example proposed in \cite{10383604}:
\small
\begin{align*}
    \begin{array}{cc}
      A=\left[\begin{array}{ccc}
         -0.53  & 0.42  & -0.44\\
         0.42  &  -0.56 & -0.65\\
         -0.44  & -0.65 & 0.35
      \end{array}  \right], &  B=\left[\begin{array}{ccc}
         0.43  & -0.82  \\
         0.53  &  -0.78 \\
         0.26  & -0.40 
      \end{array}  \right], \\ Q=\left[\begin{array}{ccc}
         6.12  & 1.72  & 0.53\\
         1.72  &  6.86& 1.72\\
         0.53  & 1.72 & 5.73
      \end{array}  \right],  & R=\left[\begin{array}{cc}
         1.15  & -0.23  \\
         -0.23  &  3.62
      \end{array}  \right].
    \end{array}
\end{align*}
\normalsize
The initial estimates $\hat{A}_0$, $\hat{B}_0$, and the matrix $H_0$ required for both ORLS+PI and OLRS+PG are set to $1.3A$, $0.7B$, and $H_0 = 0.01 I_5$, respectively. The initial feedback gain $\hat{K}_0$ is set to the optimal gain for the LQR problem associated with $(\hat{A}_0, \hat{B}_0, Q, R)$. The OLRS+PG method uses the same online policy employed in Algorithm \ref{Algo4}, with a feedback term $\hat{K}_tx_t$ plus a dithering signal $e_t \in [-10, +10]$ to ensure sufficiently informative data. The stepsize $\gamma$ of ORLS+PG is empirically set to $0.005$. Figure \ref{fig:figure2} investigates the convergence of {kernel of closed-loop evaluation} $\hat{P}_t$  by considering three different types of noise, set as \eqref{noiseset}.
\begin{figure}[H]
    \centering
    \includegraphics[width=0.5\linewidth]{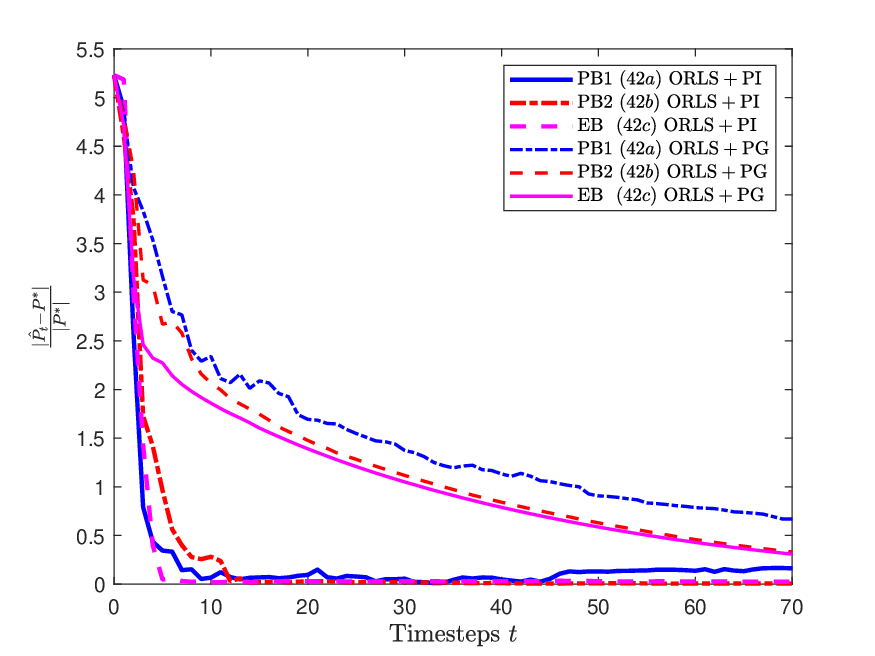}
    \caption{Comparison of ORLS+PI with ORLS+PG}
    \label{fig:figure2}
\end{figure}
As seen in Figure \ref{fig:figure2}, the ORLS+PI method exhibits faster convergence of $\hat{P}_t$ compared to the ORLS+PG methods. This is due to the nature of the PI method, which can be viewed as a Newton method. For the ORLS+PG methods, the stepsize can only be tuned empirically, and selecting an optimal stepsize to ensure convergence remains an open question. Instead, for ORLS+PI, owing to the analyses carried out in this work, there are systematic guidelines for choosing the initialization based on the bounds of the data sequence. Examining \eqref{definitionmathcalD} reveals that the upper bound also grows as the noise magnitude increases. Consequently, when the noise is larger, the initialization must be chosen closer to the true system to ensure convergence.

\section{Conclusion}\label{sec:conclusion}
In this work, we studied the application of indirect data-driven policy iteration to the LQR problems when data are subject to adversarial bounded noise. First we analyzed the convergence properties of RLS, establishing an upper bound on the estimation error. This result is meaningful for the indirect data-driven control method, as it provides guarantees on control performance by quantifying the accuracy of model estimates obtained from noisy data. Subsequently, we conceptualized the algorithm as a feedback interconnection between an identification scheme and the PI algorithm, both framed as algorithmic dynamical systems that realize concurrent learning and control. We analyzed the convergence properties of such a nonlinear closed-loop under different noise and parameters initialization scenarios to provide a comprehensive picture on the robustness of such data-driven schemes. In future work, it will be important to explore unbounded stochastic noise and investigate its impact on the performance of RLS and the coupled RLS+PI system. Additionally, we aim to explore direct data-driven policy iteration, which bypasses the system identification step and directly utilizes data to formulate the PI procedure, with a particular focus on its performance in noisy data and the relative strengths and weaknesses with respect to indirect schemes.

\section*{Acknowledgement}
  Bowen Song acknowledges the support of the International Max Planck Research School for Intelligent Systems (IMPRS-IS). Andrea Iannelli acknowledges the German Research Foundation (DFG) for support of this work under Germany’s Excellence Strategy - EXC 2075 – 390740016.
  
\appendix
\section{Technical Proof}
\subsection{Proof of Theorem \ref{theoremRLS}}\label{ProofofTheorem4}
\begin{proof}[Proof of Theorem \ref{theoremRLS}]
    From \eqref{error}, we have:
    \begin{equation}\label{errornorm}
    \begin{split}
            \lvert \Delta \theta_t \rvert &\leq a\lvert \Delta\theta_0\rvert\lvert H_t^{-1} \rvert +\left(\sum_{k=1}^{t}\lvert w_k\rvert \lvert d_k \rvert \right)\lvert H_t^{-1} \rvert\\
            &\leq a\lvert \Delta\theta_0\rvert\lvert H_t^{-1} \rvert +\bar{d}\left(\sum_{k=1}^{t}\lvert w_k\rvert \right)\lvert H_t^{-1} \rvert.          
    \end{split}
\end{equation}
With the definition of local persistency, we have:
\begin{equation*}
    \lambda_{\min}(H_t)\geq a+\lfloor \frac{t}{\lceil \frac{N_d}{M_d}\rceil M_d }\rfloor\alpha_d \geq a+\lfloor \frac{t}{M_d+N_d}\rfloor\alpha_d.
\end{equation*}
Then we have:
\begin{equation}\label{increase}
   \begin{split}
        \lvert H_t^{-1}\rvert &\leq \frac{n_x+n_u}{a+\lfloor \frac{t}{M_d+N_d}\rfloor\alpha_d} \\
        &\leq \frac{(n_x+n_u)(M_d+N_d)}{\min(a,\alpha_d)t}, ~\forall t\in \mathbb{Z}_{++}.
   \end{split}
\end{equation}
Substituting \eqref{increase} into \eqref{errornorm}, we obtain:
\begin{equation}\label{01112024}
  \begin{split}
          \lvert \hat{\theta}_t -\theta \rvert \leq \beta(\lvert \hat{\theta}_0 - \theta \rvert,t) +c\frac{\sum_{k=0}^t\lvert w_k \rvert}{t},~\forall t\in \mathbb{Z}_{++}.
  \end{split}
\end{equation}
Based on the bound defined in \eqref{Instant}, we obtain:
\begin{equation}\label{noise1}
    \sum\limits_{k=1}^{t}\lvert w_k \rvert \leq{t} \sup\limits_t \sqrt{w_t^\top w_t} \leq {t} \lVert w \rVert_\infty.
\end{equation}
Substituting \eqref{noise1} into \eqref{01112024}, we conclude the proof of Theorem \ref{theoremRLS}.
\end{proof}
\subsection{Proof of Corollary \ref{Coro1}}\label{ProofofCorollary2}
\begin{proof}
For the proof of Corollary \ref{Coro1}, we use the AM–GM inequality, 
\begin{equation}\label{noise2}
    \sum\limits_{k=1}^{t}\lvert w_k \rvert \leq \sqrt{t} \sqrt{\sum\limits_{k=1}^{t}w_k^\top w_k }\leq \sqrt{t} \sqrt{\sum\limits_{k=1}^{\infty}w_k^\top w_k }\leq \sqrt{t} \lVert w \rVert_2.
\end{equation}
Substituting \eqref{noise2} into \eqref{01112024}, we conclude the proof. 
\end{proof}
\subsection{Proof of Theorem \ref{Theorem5}}\label{ProofofTheorem5}
\begin{proof}[Proof of Theorem \ref{Theorem5}]
Based on the Assumptions \ref{asspersistentence} and \ref{ass1}, \eqref{ISSCA212} is directly proved. Now we turn to \eqref{ISSCA12}, Assumption \ref{assumption 2} guarantees the formulation of standard PI procedure. Following the same step in \cite[Appendix D6]{IJRNC} 
\begin{equation}\label{ISS111}
\begin{split}
   &\hat{P}_{t+1} =\mathcal{L}^{-1}_{\left(A,B,\hat{P}_t\right)}\left(\Gamma(\hat{P}_t)\right)+\varepsilon\left(\Delta A_t,\Delta B_t\right),
\end{split}
\end{equation}
where \begin{equation}\label{iss5}
  \begin{split}
      \varepsilon&\left(\Delta A_t,\Delta B_t\right):=-\mathcal{L}^{-1}_{\left(A,B,\hat{P}_t\right)}\left(\Gamma(\hat{P}_t))\right) \\&+\mathcal{L}^{-1}_{\left(\hat{A}_t,\hat{B}_t,\hat{P}_t\right)}\left(Q+\hat{\alpha}_t^\top \hat{\beta}_t^{-1}R\hat{\beta}_t^{-1}\hat{\alpha}_t\right),
  \end{split} 
\end{equation} and $\Gamma(\hat{P}_t)$ is defined in \eqref{Gmma}.
Using the same arguments in \cite{IJRNC}, we can prove that:
\begin{equation}\label{14112024}
    \lvert \varepsilon\left(\Delta A_t,\Delta B_t\right) \rvert \leq \bar{C} \lvert \Delta \theta_t\rvert,
\end{equation}
where  $\bar{C}$ is polynomial of $(A, B, Q, R)$. For the detailed computation steps and derivation of $\bar{C}$, we refer to \cite[Appendix D6]{IJRNC}. Then we can prove: 
\begin{equation}\label{ISS17}
  \begin{split}
     \lvert \hat{P}_{t}- P^*\rvert &\leq c \lvert \hat{P}_{t-1}- P^*\rvert+\bar{C} \lvert \Delta \theta_t\rvert\\
     &\leq c^{t}\lvert \hat{P}_{0}- P^*\rvert+\bar{C}\left(1+c+...+c^{t-1}\right)\lVert \Delta\theta \rVert_{\infty}\\
     &\leq c^{t}\lvert \hat{P}_{0}- P^*\rvert+ \frac{\bar{C}}{1-c}\lVert \Delta\theta \rVert_{\infty}.
  \end{split}
\end{equation}
Then we conclude the proof of \eqref{ISSCA12}.
\end{proof}
\subsection{Proof of Theorem \ref{theorem6}}\label{ProofofTheorem6}
\begin{proof}[Proof of Theorem \ref{theorem6}]
    In this proof, the robustness of PI algorithms plays a central role, as outlined in our previous work \cite[Theorem 7]{ECC}. For clarity and completeness, we recall this theorem here:
    \begin{Theorem}[Robustness of PI \cite{ECC}]\label{TheoremRPI}
    Given $\sigma$ and $\delta_1$ defined in Theorem \ref{Theorem2}, there always exist constants $\Bar{a}_{p}(\delta_1,\sigma)\geq 0$ and $\Bar{b}_{p}(\delta_1,\sigma)\geq 0$ such that if $\|a\|_\infty\leq \Bar{a}_{p}$, $\|b\|_\infty\leq \Bar{b}_{p}$ and $\hat{P}_0\in \mathcal{B}_{\delta_1}(P^*)$, where sequences $\{a_t\}$ and $\{b_t\}$ are defined as\begin{equation}\label{251020242}
    a_t:=\lvert \Delta A_t\rvert,~ b_t:=\lvert \Delta B_t\rvert,
\end{equation} with $\Delta A_t:= \hat{A}_t-A,~\Delta B_t:= \hat{B}_t-B$, then 
    \begin{enumerate}
        \item $\hat{K}_t$ is stabilizing, $\forall t \in \mathbb{Z}_{+}$;
        \item the following holds,:
        \begin{equation}\label{ISSProof}
        \begin{split}
            |\hat{P}_{t} - P^*| &\leq \beta_{p}(|\hat{P}_{0} - P^*|, t) 
            +\gamma_{1}(\|a\|_{\infty}) \\
            &\quad + \gamma_{2}(\|b\|_{\infty}){\leq \delta_1},~\forall t \in \mathbb{Z}_{+},
        \end{split}
        \end{equation}
        where $\beta_{p}(x, t) := \sigma^{t}x$; $\gamma_{1}(x) := \frac{\bar{p}_a}{1-\sigma}x$; $\gamma_{2}(x) := \frac{\bar{p}_b}{1-\sigma}x$ with constants $\bar{p}_a, \bar{p}_b >0$;
        \item if $\lim\limits_{t \to \infty} |\Delta A_t| = 0 $ and $\lim\limits_{t \to \infty} |\Delta B_t| = 0 $, then $\lim\limits_{t \to \infty} |\hat{P}_{t} - P^*| = 0 $.
    \end{enumerate}
\end{Theorem}
To proceed with the proof, we first verify the conditions under which Theorem \ref{TheoremRPI} holds.

{From Theorem \ref{TheoremRPI}, if $\|a\|_\infty\leq \Bar{a}_{p}$, $\|b\|_\infty\leq \Bar{b}_{p}$ and $\hat{P}_0=P^*$, ensuring that the conditions of Theorem \ref{TheoremRPI} hold, then we have $|\hat{P}_{t} - P^*| \leq \delta_1,~\forall t \in \mathbb{Z}_{+}$. Moreover, this guarantees that $\lim\limits_{t\rightarrow\infty}|\hat{P}_{t} - P^*| \leq \delta_1$. Now, consider fixed matrices $\tilde{A}$ and $\tilde{B}$ satisfying $|\tilde{A}-A| \leq \bar{a}_{p}$ and $|\tilde{B}-B| \leq \bar{b}_{p}$, then given an initial condition $\hat{P}_0=P^*$, we conclude that $\lim\limits_{t\rightarrow\infty}|\hat{P}_{t} - P^*|= |{P}^*_{(\tilde{A},\tilde{B})} - P^*|\leq \delta_1$, where ${P}^*_{(\tilde{A},\tilde{B})}$ is the optimal solution to \eqref{Kpolicyimprovement_EQ} corresponding to $(\tilde{A},\tilde{B},Q,R)$. Thus, we can conclude, for any $\tilde{A}$ and $\tilde{B}$ satisfying $|\tilde{A}-A| \leq \bar{a}_{p}$ and $|\tilde{B}-B| \leq \bar{b}_{p}$, then  ${P}^*_{(\tilde{A},\tilde{B})}\in \mathcal{B}_{\delta_1}(P^*)$.} 

When the maximum estimation error $\overline{\Delta\theta}(\hat{\theta}_0, \bar{D})\leq \min\{ \bar{a}_p,\bar{b}_p\}$, then we have 
\begin{equation}\label{11120241}
\begin{split}
        \gamma_{1}(\|a\|_{\infty}) + \gamma_{2}(\|b\|_{\infty}) &= \frac{\bar{p}_a\|a\|_{\infty}+\bar{p}_b\|b\|_{\infty}}{1-\sigma}\\
        &\leq \frac{(\bar{p}_a+\bar{p}_b)\|\Delta\theta\|_{\infty}}{1-\sigma}.
\end{split}
\end{equation}
{Together with \eqref{max1}, when Assumption 4 holds and the initial policy $\hat{K}_0$ is selected as the solution to the \eqref{Kpolicyimprovement_EQ} using $(\hat{A}_0,\hat{B}_0,Q,R)$, the conditions required by Theorem \ref{theorem6} are satisfied.} 
Substituting \eqref{11120241} into \eqref{ISSProof}, we conclude \eqref{ISSCA}. Now we turn to prove \eqref{ISSCA2}. If the data sequence $\{d_t\}$ is bounded, then we can directly use Theorem \ref{theoremRLS} to prove \eqref{ISSCA2}. 

For matrices, $\lvert \cdot\rvert_2$ denotes their induced-2 norm. Based on Theorem \ref{TheoremRPI}, $\hat{K}_t$ is stabilizing, for all $t \in \mathbb{Z}_+$. Then we can define:
\begin{equation}\label{defineKCL}
    \bar{K}_\mathrm{cl}:=\sup\limits_{\begin{array}{ccc}
         \lvert \hat{A}-A \rvert \leq \bar{a}_p,\\
         \lvert \hat{B}-B \rvert \leq \bar{b}_p,  \\
         P\in \mathcal{B}_{\delta_1}(P^*)
    \end{array}}\lvert \hat{A}+\hat{B} (\hat{B}^\top P \hat{B}+R)^{-1}\hat{B}^\top P \hat{A}\rvert_2
\end{equation}
and we have $\bar{K}_\mathrm{cl} \in [0,1)$. 
The additional excitation term $e_t$ satisfies $\lVert e_t\rVert \leq \Bar{e},~\forall~t \in\mathbb{Z}_+$ \eqref{boundede}. Additionally, we have
\begin{equation}\label{boundedK}
\begin{split}
        \lvert \hat{K}_t \rvert &=\lvert (R+B_t^\top \hat{P}_tB_t)^{-1}\hat{B}_t^\top\hat{P}_t\hat{A}_t  \rvert\\
        &\leq \underbrace{\lvert R^{-1}\rvert(\lvert{B}\rvert+\lVert\Delta\theta\rVert_\infty)(\lvert P^*\rvert +\delta_1) (\lvert{A}\rvert+\lVert\Delta\theta\rVert_\infty)}_{=:\bar{K}} 
\end{split}
\end{equation}
Then we can introduce the following lemma, which shows the boundedness of $x_t$:
\begin{Lemma}[Boundedness of state $x_t$] \label{Boundness}
    Given the system \eqref{LTI} with noise satisfying \ref{Instant} and with the control input $u_t=\hat{K}_tx_t+e_t$, where ${\hat{K}}_t$ is the stabilizing gain from ORLS+PI and $e_t$ satisfies \eqref{boundede}, the state of system \eqref{LTI} remains bounded:
     \begin{equation}
         \lvert x_t \rvert \leq \max\left(\frac{\lvert B \rvert \Bar{e}+\lVert w \rVert_\infty}{1-\bar{K}_\mathrm{cl}},\lvert x_0 \rvert\right)=:\Bar{x},~\forall t\in\mathbb{Z}_+,
     \end{equation}
     where $\bar{K}_\mathrm{cl}$ is defined in \eqref{defineKCL} and $\bar{e}$ is defined in \eqref{boundede}.
\end{Lemma}
\begin{proof} [Proof of Lemma \ref{Boundness}]
    For the case $\lvert x_t \rvert \geq \frac{\lvert B \rvert \Bar{e}+\lVert w \rVert_\infty}{1-\bar{K}_\mathrm{cl}}$,
    \begin{equation*}
    \begin{split}
                \lvert x_{t+1} \rvert &\leq \lvert A+B{\hat{K}}_t \rvert_2 \lvert x_t \rvert+\lvert B \rvert \Bar{e}+\lVert w \rVert_\infty \\
                &\leq \bar{K}_\mathrm{cl}  \frac{\lvert B \rvert \Bar{e}+\lVert w \rVert_\infty}{1-\bar{K}_\mathrm{cl}}+\lvert B \rvert \Bar{e}+\lVert w \rVert_\infty\\
                &=\frac{\lvert B \rvert \Bar{e}+\lVert w \rVert_\infty}{1-\bar{K}_\mathrm{cl}}.
    \end{split}
    \end{equation*}
    Together with the upper bound on the initialization, we conclude the proof. 
\end{proof}
Further, we can also derive the bound of the data $d_t$:
\begin{Lemma}[Boundedness of data $d_t$] \label{Boundness2}
    Given the system \eqref{LTI} with noise satisfying \eqref{Instant} and with the control input $u_t=\hat{K}_tx_t+e_t$ where ${\hat{K}}_t$ is the stabilizing gain from ORLS+PI and $e_t$ satisfies \eqref{boundede}, the data $d_t$, which is employed for RLS, is bounded:
     \begin{equation}\label{definitionmathcalD}
     \begin{split}
          \lvert d_t \rvert &= \left\lvert\left[\begin{array}{cc}
         x_t  \\
         u_t 
    \end{array}\right]\right\rvert\leq \left\lvert\left[\begin{array}{cc}
         I  \\
         {\hat{K}}_t 
    \end{array}\right]\right\rvert\lvert x_t \rvert +\left\lvert\left[\begin{array}{cc}
         0  \\
         e_t 
    \end{array}\right]\right\rvert\\
    &\leq (1+\bar{K})\bar{x}+\bar{e}=:\bar{D}
     \end{split} 
     \end{equation}
     where $\bar{K}$ is defined in \eqref{boundedK} and $\bar{x}$ is defined in Lemma \ref{Boundness}.
\end{Lemma}
Using the upper bound of the data sequence $\{d_t\}$ and together with Theorem \ref{theoremRLS}, we can conclude \eqref{ISSCA2}.
\end{proof}

\bibliographystyle{unsrt}  
\bibliography{references}

\end{document}